\documentclass[11pt,a4paper]{amsart}
\usepackage{fullpage}
\usepackage{hyperref,amsmath,amsthm,amssymb,enumerate,amsfonts,graphicx}
\usepackage{color}
\usepackage[numbers,sort&compress]{natbib}
\usepackage{algorithm,algorithmic}

\renewcommand{\l}{\ell}

\newcommand{\DEF}{\sl}    

\theoremstyle{plain}
\newtheorem{theorem}{Theorem}[section]
\newtheorem{lemma}[theorem]{Lemma}
\newtheorem{corollary}[theorem]{Corollary}
\newtheorem{conjecture}[theorem]{Conjecture}

\newcommand{\bep}{Balanced Minimum Evolution Problem}

\begin{document}
\title{Approximating the Balanced Minimum Evolution Problem}
\thanks{This work was supported by the ``Actions de Recherche Concert\'ees'' (ARC) fund of the ``Communaut\'e fran\c{c}aise de Belgique''. G.J. is a Postdoctoral Researcher of the ``Fonds National de la Recherche Scientifique'' (F.R.S.--FNRS)}

\author{Samuel Fiorini}
\address{\newline D\'epartement de Math\'ematique
\newline Universit\'e Libre de Bruxelles
\newline Brussels, Belgium}
\email{sfiorini@ulb.ac.be}

\author{Gwena\"el Joret}
\address{\newline D\'epartement d'Informatique
\newline Universit\'e Libre de Bruxelles
\newline Brussels, Belgium}
\email{gjoret@ulb.ac.be}

\maketitle

\begin{abstract}
We prove a strong inapproximability result for the \bep. Our proof also implies that the problem remains NP-hard even when restricted to metric instances. Furthermore, we give a MST-based $2$-approximation algorithm for the problem for such instances.
\end{abstract}

\section{Introduction}

Let $[n] := \{1,\ldots,n\}$ be a set of $n$ {\DEF species}. Let $(\delta_{ij})$ be a $n\times n$ symmetric matrix with nonnegative entries and zeroes on the diagonal, where $\delta_{ij}$ represents the {\DEF dissimilarity} between species $i$ and $j$. The {\DEF \bep} is to find a {\DEF cubic} tree $T$ (every internal vertex has degree 3) with $n$ leaves, together with a bijection between the leaves of $T$ and the $n$ species, so that
\begin{equation}
\label{eq:bep_objective}
f(T) := \sum_{i \neq j} \delta_{ij} 2^{1-d_{ij}} = \sum_{i < j} \delta_{ij} 2^{2-d_{ij}} 
\end{equation}
is minimized, where $d_{ij}$ denotes the distance between the leaves for species $i$ and $j$ in $T$. We point out that our objective function is twice the {\DEF length} of the tree $T$, which is the commonly used objective function.

This computational biology problem was introduced by Desper and Gascuel \cite{DesperGascuel02}, inspired by work of Pauplin \cite{Pauplin00}, and has been studied, e.g., in \cite{DesperGascuel04,SempleSteel04,Eickmeyer_et_al08,CuetoMatsen10,%
Catanzaro_et_al11}. Although no hardness proof for the problem has been published, 
it appears that it was known to be NP-hard since 2004 (Guillemot~\cite{Guillemot11}). 
To our knowledge, plain NP-hardness is the strongest hardness result known about the problem. In particular, the complexity of the \bep{} is still open in case the dissimilarities are restricted to be 0/1, or to satisfy the triangle inequality. Furthermore, the problem is not known to be hard to approximate.

First, in Section \ref{sec:prelim}, we start with preliminaries. Then, in Section \ref{sec:reduct}, we prove that the \bep{} does not admit any interesting approximation algorithm (unless P $=$ NP): the problem is NP-hard to approximate to within a $c^{n}$-factor for some constant $c > 1$. Finally, in Section \ref{sec:approx}, we give a simple $2$-approximation algorithm for the problem, in case the dissimilarities $\delta_{ij}$ satisfy the triangle inequality. By results of the previous section, the problem is NP-hard in this case.

\section{Preliminary Remarks and Observations}
\label{sec:prelim}

\subsection{Kraft's Inequality}

Kraft's inequality for a binary tree with $n$ leaves states that
\begin{equation}
\label{eq:Kraft}
\sum_{i\in [n]} 2^{-d_{i}} \leqslant 1,
\end{equation}
where $d_{i}$ is the distance from the root $r$ to the $i$th leaf. It is easy to prove that, if the tree is a binary cubic tree (meaning that all internal vertices have degree 3 except the root), then equality holds in \eqref{eq:Kraft}. This implies that for every feasible solution $T$ to the \bep{} and for every fixed leaf $j$, 
\begin{equation}
\label{eq:Kraft_bep}
\sum_{i \in [n] \atop i \neq j} 2^{2-d_{ij}} = 2.
\end{equation}

\subsection{The Objective as an Average Over Compatible Tours}

A result of Semple and Steel~\cite{SempleSteel04} states that $2^{2-d_{ij}}$ is the probability that leaves $i$ and $j$ are consecutive in a (undirected) tour on the leaves of $T$ chosen uniformly at random from the tours {\DEF compatible} with $T$, that is, such that the tree $T$ can be embedded in the plane so that the tour visits the leaves of $T$ in clockwise order. Thus, we have the following lemma.

\begin{lemma}
\label{lem:random_tour}
For all feasible solutions $T$ to the \bep, $f(T)$ is the expected cost of a random tour compatible with $T$.
\end{lemma}

In the light of Lemma \ref{lem:random_tour}, it should not surprise the reader that one can define the \bep{} over \emph{all} trees $T$ with $n$ leaves, by defining $f(T)$ as the expected cost of a tour picked uniformly at random from the tours compatible with $T$. Then, letting $P_{ij}=P_{ij}(T)$ denote the unique $i$--$j$ path in $T$ (with vertex set $V(P_{ij})$) and
\begin{equation}
\label{eq:pi_vector}
\pi_{ij} := 2 \prod_{u \in V(P_{ij}) \atop u \neq i, j} \frac{1}{\deg_T(u)-1},
\end{equation}
one has
$$
f(T) = \sum_{i < j} \delta_{ij} \pi_{ij}.
$$
However, it is known that the $\pi$-matrices of non-cubic trees are convex combinations of the $\pi$-matrices of cubic trees \cite{Eickmeyer_et_al08}, hence for every non-cubic tree $T$ there always exists a cubic tree $T'$ with $f(T') \leqslant f(T)$. In Section \ref{sec:approx}, we will give a polynomial time algorithm to find such a cubic tree $T'$.

\subsection{The ``All 1'' Case}

Every solution is optimal in that case:

\begin{lemma}
\label{lem:all_1}
Suppose $\delta_{ij}=1$ for every $i,j$ with $i\neq j$. Then, for all feasible solutions $T$, 
$$
f(T) = n.
$$
In particular, the optimum of the {\bep} is $n$.
\end{lemma}
\begin{proof}[Proof 1]
By \eqref{eq:Kraft_bep},
$$
f(T) = \sum_{i < j} 2^{2-d_{ij}} = \frac{1}{2} \sum_{j \in [n]} 
\sum_{i \in [n] \atop i \neq j} 2^{2-d_{ij}} = \frac{1}{2} 2n = n
$$
\end{proof}
\begin{proof}[Proof 2]
Every tour on the leaves of $T$ has $n$ edges, of cost $1$ each. By Lemma \ref{lem:random_tour}, it follows that $f(T) = n$ for all feasible solutions $T$.
\end{proof}

\section{NP-Hardness and Inapproximability}
\label{sec:reduct}

\begin{theorem}
\label{th:inapprox}
There exists a constant $c > 1$ such that the {\bep} has no
$c^{n}$-approximation algorithm unless P $=$ NP, where
$n$ denotes the number of species. This remains true even when all entries of the dissimilarity matrix are in $\{0,1\}$.
\end{theorem}
\begin{proof}
The reduction is from the $3$-Colorability Problem: We are given a (simple, undirected)  
graph $G$ on $p$ vertices, and have to decide whether $V(G)$ can be partitioned into three stable sets
(recall that a {\em stable set} is a set of mutually non-adjacent vertices).

We may assume without loss of generality that $G$ contains two vertex-disjoint triangles. Indeed, if not
it suffices to add twice three new vertices to $G$ that form a triangle; this has clearly no influence on whether $G$ is $3$-colorable or not. 

Let $\lambda$ be an arbitrary constant with $1/2 < \lambda < 2/3$. 
We will prove the claim with 
$$
c:= 2^{(2/3 - \lambda)(3 - 4\lambda)} > 1.
$$
(By taking $\lambda$ sufficiently close to $1/2$, one has $c \geqslant 1.12$.)

Let $m$ be the number of edges in $G$.
We may assume 
\begin{equation}
\label{eq:m}
m \leqslant 2^{(2/3 - \lambda) p} = 2^{(2/3 - \lambda) |V(G)|}
\end{equation}
because otherwise $G$ has bounded size and we can check whether $G$ is $3$-colorable using brute force.

Define $k$ as the smallest integer satisfying 
$k \geqslant p/(2\lambda - 1)$ and $k \equiv 1$~(mod 3).
Consider an arbitrary ordering $v_{1}, v_{2}, \dots, v_{p}$ of the vertices of $G$.
We define an instance of the {\bep} with $n:= p + k$ species as follows.
The first $p$ species are associated with the
vertices of $G$: species $i$ (for $i\in [p]$) corresponds to vertex $v_{i}$. 
The matrix $(\delta_{ij})$ is defined by setting, for $i \neq j$, 
$$
\delta_{ij} := \left \{
\begin{array}{ll}
1 & \quad \text{if $i, j\in [p]$ and $v_{i}v_{j} \in E(G)$}, \\
0 & \quad \text{otherwise}. \\
\end{array}
\right.
$$

Consider an optimal solution for the instance of the {\bep} described above.
This solution is a cubic tree $T$ with $n$ leaves together with
a bijection from the set of species to the set of leaves of $T$.
For simplicity, we denote by $v_{i}$ ($i \in  [n]$) the leaf of $T$
associated to species $i$. (Thus, when $i \leqslant p$, 
$v_{i}$ denotes both the $i$th vertex of $G$ and the corresponding leaf of $T$; 
which one is meant will be clear from the context.)
The cost of this optimal solution is denoted $OPT$. Thus, we have
\begin{equation}
\label{eq:OPT}
OPT = \sum_{v_{i}v_{j} \in E(G)}  2^{2 - d_{ij}}
\end{equation}
where $d_{ij}$ is the distance between species $i$ and $j$ in $T$.

First we show:
\begin{equation}
\label{eq:dist}
\textrm{If $d_{ij} > \lambda k$ for all $v_{i}v_{j} \in E(G)$
then $G$ is $3$-colorable}.
\end{equation}
Consider an arbitrary triangle in $G$; without loss of generality we may assume
that the vertices of this triangle are $v_{1}, v_{2}, v_{3}$. 
Let $C$ be the union of the $v_{1}$--$v_{2}$ path, the $v_{1}$--$v_{3}$ path, and
the $v_{2}$--$v_{3}$ path in $T$. 
(Recall that there is unique path between two given vertices in a tree, thus $C$ is well defined.)
Then $C$ is isomorphic to a subdivision of the claw $K_{1, 3}$ and its three leaves are 
$v_{1}, v_{2}$, and $v_{3}$. Let $w$ be the unique vertex in $C$ with degree $3$.
Let $P_{\l}$ ($\l\in \{1,2,3\}$) denote the path obtained from 
the $v_{\l}$--$w$ path in $T$ by removing
$w$. 
Since $v_{1}, v_{2}, v_{3}$ are pairwise adjacent in $G$, we have
\begin{equation}
\label{eq:P}
|P_{\l}| + |P_{\l'}| = d_{\l\l'} > \lambda k
\end{equation}
for all $\l, \l' \in \{1,2,3\}$ with $\l \neq \l'$. 
($|P_{\l}|$ stands for the number of vertices in $P_{\l}$.)

Let $T_{\l}$ ($\l\in \{1,2,3\}$) be the component of $T - w$ containing $v_{\l}$, and let
$X_{\l}$ be the set of internal vertices in $T$ that are included in $T_{\l}$.
Observe that $T$ has $n-2$ internal vertices and that
all vertices of $P_{\l}$ ($\l\in \{1,2,3\}$) are internal vertices of $T$, except for $v_{\l}$.
Using~\eqref{eq:P} and $k \geqslant p/(2\lambda - 1)$ we obtain
\begin{equation}
\label{eq:X}
|X_{\l}| \leqslant (n - 2) + 1 - \sum_{\l' \in \{1,2,3\}, \l' \neq \l} |P_{\l'}| 
<  n - \lambda k - 1
= (1 - \lambda)k + p - 1
\leqslant \lambda k  - 1
\end{equation}
for all $\l \in \{1,2,3\}$.
 
Let $S_{\l}$ ($\l\in \{1,2,3\}$) be the set of vertices $v_{i}$ of $G$ such that
$v_{i}$ is a leaf of $T$ that is included in $T_{\l}$. 
(Thus $S_{1} \cup S_{2} \cup S_{3} = V(G)$.)
Every two vertices in $S_{\l}$ are at distance at most 
$|X_{\l}| + 1 < \lambda k$ in $T$ by \eqref{eq:X}.
Therefore, $S_{1}, S_{2}, S_{3}$ are stable sets of $G$, and $G$ is $3$-colorable.
This proves~\eqref{eq:dist}.

Next we prove:
\begin{enumerate}[(a)]
\item if $G$ is not $3$-colorable then $\displaystyle OPT \geqslant 2^{2 - \lambda k}$;
\item if $G$ is $3$-colorable then $\displaystyle OPT \leqslant m\cdot 2^{2 - (2k+4)/3}$.
\end{enumerate}
The first part of the above claim is a direct consequence of~\eqref{eq:dist}:
If $G$ is not $3$-colorable then there is an edge $v_{i}v_{j}$ of $G$ such that
$d_{ij} \leqslant \lambda k$, and hence
$OPT \geqslant 2^{2 - d_{ij}} \geqslant 2^{2 - \lambda k}$.

For the second part, let $S_{1}$, $S_{2}$, $S_{3}$ denote the three color classes of a 
$3$-coloring of $G$. Recall that $G$ has two vertex-disjoint triangles, which implies
$|S_{\l}| \geqslant 2$ for every $\l \in \{1,2,3\}$. 
We build a feasible solution $T'$ from this coloring which will imply
the desired upper bound on $OPT$.

The tree $T'$ is defined as follows. First, for each $\l \in \{1,2,3\}$, create
a path $P_{\l}$ on $(k-1)/3 + |S_{\l}| - 1$ vertices (here we use that $k \equiv 1$~(mod 3)).
Let $a_{\l}$ and $b_{\l}$ be the two endpoints of $P_{\l}$.
Create a new vertex $w$ and make it adjacent to $b_{1}, b_{2}$, and $b_{3}$.
Next, attach a leaf to each vertex of degree $2$ in the resulting tree, and attach 
two new leaves to each of $a_{1}$, $a_{2}$, and $a_{3}$. This defines the tree $T'$.
The species are placed in the following way on the leaves of $T'$:
for each $\l \in \{1,2,3\}$, put the species corresponding to vertices in $S_{\l}$ on
the $|S_{\l}|$ leaves that are closest to $a_{\l}$ in $T'$, in an arbitrary way.
The remaining $k$ species are placed arbitrarily on the $k$ leaves of $T'$ that remain free.

Since $|S_{\l}| \geqslant 2$ for every $\l \in \{1,2,3\}$,
the two leaves adjacent to $a_{\l}$ are associated with species in $[p]$, and thus
the first $(k-1)/3$ vertices of the  path from $b_{\l}$ to $a_{\l}$ in $T'$ 
are adjacent to leaves associated with species not in $[p]$.
Hence, if $i,j \in [p]$ are species such that $v_{i}v_{j}\in E(G)$, then
they are at distance at least $(k-1)/3 + 1 + (k-1)/3 + 1 = (2k+4)/3$ in $T'$.
Therefore, the cost of this feasible solution is at most $m\cdot 2^{2 - (2k+4)/3}$, implying
$OPT \leqslant m\cdot 2^{2 - (2k+4)/3}$ as claimed.

Now, since
\begin{align*}
\frac{2^{2 - \lambda k}}{m\cdot 2^{2 - (2k+4)/3}} 
&= \frac{2^{(2/3 - \lambda) k  + 4/3}}{m} & \\
&> \frac{2^{(2/3 - \lambda) k}}{m} & \\
&\geqslant 2^{(2/3 - \lambda) (k - p)} & \textrm{(by \eqref{eq:m})} \\
&= 2^{(2/3 - \lambda) (n - 2p)} & \\
&\geqslant 2^{(2/3 - \lambda)(3 - 4\lambda)n} 
& \textrm{(since $p \leqslant (2\lambda - 1) k \leqslant (2\lambda - 1) n$)} \\
&= c^n, & & 
\end{align*}
it follows that a $c^{n}$-approximation algorithm for the {\bep} could
be used to decide whether $G$ is $3$-colorable or not. This concludes the proof.
\end{proof}

An instance of the {\bep} is said to be {\DEF metric} if the dissimilarity matrix $(\delta_{ij})$ is a semimetric, that is, if the $\delta_{ij}$'s satisfy
$$
\delta_{ik} \leqslant \delta_{ij} + \delta_{jk}
$$
for all distinct species $i$, $j$, $k$. 

\begin{corollary}
The {\bep} is NP-hard on metric instances. This remains true even if the non-diagonal entries of the dissimilarity matrix are all in $\{1,2\}$.
\end{corollary}
\begin{proof}
By Theorem~\ref{th:inapprox} the {\bep} is NP-hard when all dissimilarities are in $\{0,1\}$. Consider such an instance and add $1$ to every non-diagonal entry of the dissimilarity matrix $(\delta_{ij})$, giving a dissimilarity matrix $(\delta'_{ij})$. Then $(\delta'_{ij})$ is a semimetric, because
$$
\delta'_{ik} \leqslant 2 = 1 + 1 \leqslant \delta'_{ij} + \delta'_{jk}
$$
for all distinct species $i$, $j$, $k$. 

Consider a feasible solution $T$ to the instance, and let $f(T,(\delta_{ij}))$ and $f(T,(\delta'_{ij}))$ denote the cost of the solution w.r.t.\ $(\delta_{ij})$ and $(\delta'_{ij})$, respectively.
Let $(u_{ij}) := (\delta'_{ij}) - (\delta_{ij})$.
Then $f(T,(\delta'_{ij})) 
= f(T,(\delta_{ij})) + f(T,(u_{ij}))
= f(T,(\delta_{ij})) + n$ by Lemma \ref{lem:all_1}. 
It follows that a solution to the modified instance is optimal
if and only if it is optimal for the original instance. 
\end{proof}

\section{A $2$-Approximation Algorithm for Metric Instances}
\label{sec:approx}

In this section we assume that the dissimilarity matrix $(\delta_{ij})$ is a semimetric. 
We describe a MST-based $2$-approximation algorithm for this special case.

\subsection{Two Lower Bounds} 

Let $TSP$ denote the cost of an optimal tour on the $n$ species with respect to the costs $\delta_{ij}$, and let again $OPT$ denote the cost of an optimal solution to the \bep. By Lemma \ref{lem:random_tour}, because the average of a random variable is always at least the minimum value achieved by the random variable, we conclude
$$
OPT \geqslant TSP.
$$
Now let $MST$ denote the cost of a minimum spanning tree on the species w.r.t.\ the costs $\delta_{ij}$. It is known that $MST$ is a lower bound on $TSP$, thus also
\begin{equation}
\label{eq:MST_bd}
OPT \geqslant MST.
\end{equation}

\subsection{The Algorithm and its Analysis}

\begin{algorithm}[ht]
\caption{A $2$-approximation algorithm for metric instances.\label{alg:approx}}
\begin{algorithmic}[1]
\STATE Compute a minimum spanning tree $T_0$ on the $n$ species w.r.t.\ costs $\delta_{ij}$.
\STATE $T \longleftarrow T_0$
\WHILE{there is a species $i \in V(T)$ that is not a leaf}
\STATE Relabel internal vertex $i$ as $i'$.
\STATE Add new leaf to $T$ adjacent to $i'$ through a new edge of zero cost, label the leaf $i$.
\ENDWHILE 
\STATE
Find a feasible cubic tree $T'$ with $f(T') \leqslant f(T)$. 
\RETURN $T'$
\end{algorithmic}
\end{algorithm}

Consider Algorithm \ref{alg:approx} above. 

First, it is clear that the cost of $T_0$, as a solution of the minimum spanning tree problem, is $MST$. 

Second, observe that the modifications performed on $T$ in steps 3--6 induce an extended semimetric $(\hat{\delta}_{ij})$ defined over the whole vertex set of the final tree $T$. In this semimetric, for every leaf $i$ that was moved to the exterior of the tree, we have $\hat{\delta}_{ii'} = 0$. 

Third, observe that the final tree $T$ is an optimal solution of the minimum spanning tree problem with respect to the extended semimetric $(\hat{\delta}_{ij})$, of cost $MST$. Hence, every closed walk that visits each edge of $T$ twice has cost $2MST$. Since any tour on the leaves of $T$ that is compatible with $T$ can be obtained by shortcutting such a closed walk, every such tour has cost at most $2MST$, because $(\hat{\delta}_{ij})$ is a semimetric. 

Fourth, by combining Lemma \ref{lem:random_tour} and \eqref{eq:MST_bd}, we conclude that Algorithm \ref{alg:approx} returns a feasible solution $T'$ whose cost is at most $2MST$, hence at most $2OPT$. It follows from Lemma \ref{lem:round_to_cubic} below that the whole algorithm, and in particular step 7, can be implemented so that its running time is polynomial. 

\begin{lemma}
\label{lem:round_to_cubic}
Let $T$ be any tree with $n$ leaves, namely, the $n$ species. Then one can find in polynomial time a feasible cubic tree $T'$ with $f(T') \leqslant f(T)$.
\end{lemma}

\begin{proof}[Proof] 
Pick an internal vertex $u$ with degree $q > 3$. Next, pick two neighbors $v_1$ and $v_2$ of $u$. Let $T^{v_1v_2}$ denote the tree obtained from $T$ by adding a new internal vertex $u'$ with neighborhood $\{u,v_1,v_2\}$ and deleting $v_1$ and $v_2$ from the neighborhood of $u$. We claim that the $\pi$-matrix of $T$, as defined by \eqref{eq:pi_vector}, can be obtained as a convex combination of the $\pi$-matrices of the trees $T^{v_1v_2}$, where $v_1, v_2 \in N_T(u)$. In particular, there exists a pair $v_1$, $v_2$ such that $f(T^{v_1v_2}) \leqslant f(T)$. The lemma follows from the claim.

In order to prove the claim, denote by $(\pi_{ij})$ the $\pi$-matrix of $T$ and by $(\pi^{v_1v_2}_{ij})$ the $\pi$-matrix of $T^{v_1v_2}$. Consider a pair $i$, $j$ of leaves of $T$. 

If $u \notin P_{ij}(T)$, then $\pi^{v_1v_2}_{ij} = \pi_{ij}$ always.

Otherwise, $u \in P_{ij}(T)$. Let $n_{u}$, $n_{u'}$, $n_{uu'}$ denote the number of pairs $v_1$, $v_2$ such that $P_{ij}(T^{v_1v_2})$ contains, respectively, $u$ and not $u'$, $u'$ and not $u$, both $u$ and $u'$. Then $n_{u} = 1$, $n_{uu'} = 2(q-2)$ and $n_{u'} = {q \choose 2} - n_{u} - n_{u'} = \frac{1}{2}q^2 - \frac{5}{2}q + 3$. 

Therefore,
\begin{eqnarray*}
\sum_{\{v_1,v_2\} \subseteq N(u)} \frac{1}{{q \choose 2}} \pi^{v_1v_2}_{ij} 
&= & \frac{1}{{q \choose 2}} \left(\frac{q-1}{2} n_{u} + \frac{q-1}{q-2} n_{u'} + \frac{q-1}{2(q-2)} n_{uu'}\right) \pi_{ij}\\
&= & \frac{2}{q(q-1)} \left(\frac{q-1}{2} + \frac{q-1}{q-2} \left( \frac{1}{2}q^2 - \frac{5}{2}q + 3 \right) + \frac{q-1}{2(q-2)} 2(q-2)\right) \pi_{ij}\\
&= & \frac{2}{q} \left(\frac{1}{2} + \frac{1}{q-2} \left( \frac{1}{2}q^2 - \frac{5}{2}q + 3 \right) + 1\right) \pi_{ij}\\
&= & \frac{2}{q} \frac{1}{q-2} \left(\frac{q-2}{2} + \frac{1}{2}q^2 - \frac{5}{2}q + 3 + q-2\right) \pi_{ij}\\
&= & \frac{2}{q} \frac{1}{q-2} \left(\frac{1}{2}q^2 - q \right) \pi_{ij}\\
&= & \pi_{ij}.
\end{eqnarray*}

From what precedes, we infer that
$$
\sum_{\{v_1,v_2\} \subseteq N(u)} \frac{1}{{q \choose 2}} \pi^{v_1v_2}_{ij} = \pi_{ij}
$$
for all pairs of leaves $i$, $j$. The claim, and the result follow.
\end{proof}

Our final result follows.

\begin{theorem}
Algorithm \ref{alg:approx} is a $2$-approximation algorithm for the \bep.
\end{theorem}

\subsection{Tightness of the Lower Bounds}

Consider the metric instances of the \bep{} with $\delta_{1i} = 1$ for all $i > 1$ and $\delta_{ij} = 2$ for all pairs such that $i > 1$ and $j > 1$. For these instances, $MST = n-1$. However, as it can be easily checked, we also have $OPT \geqslant 2n-2$. Hence, $\lim_{n \to \infty} \frac{OPT}{MST} = 2$ for this family of instances. This indicates that analyzing the approximation factor of an algorithm in terms of $MST$ cannot yield a factor smaller than $2$. 

We believe that this is even true when the stronger bound $TSP$ is used, and make the following conjecture (backed by experimental evidence). 

\begin{conjecture}
The family of instances of the \bep{} in which $(\delta_{ij})$ is the shortest path metric of a $n$-vertex cycle satisfies $\lim_{n \to \infty} \frac{OPT}{TSP} = 2$.
\end{conjecture} 

\section*{Acknowledgements}
We thank Daniele Cantazaro for drawing our attention to the \bep{}, and for interesting discussions on this topic. We also thank Dirk Oliver Theis for his involvement in the early stage of the research that lead to this paper.

\bibliographystyle{plain}
\bibliography{bmep}  

\end{document}